\documentclass[12pt,journal,draftcls,onecolumn]{IEEEtran}
\usepackage{bbm}
\usepackage[dvips]{graphicx}
\usepackage{epsfig}
\usepackage[cmex10]{amsmath}
\usepackage{amssymb}
\usepackage{amsthm}
\usepackage{amsfonts}
\usepackage{bm}
\usepackage{cite}
\usepackage{hyperref}
\usepackage[svgnames]{xcolor} 
\usepackage{pstricks,pst-node,pst-plot,pstricks-add}
\usepackage[tight,footnotesize]{subfigure}
\usepackage{algpseudocode}
\usepackage{algorithm}
\usepackage{amsthm}
\theoremstyle{theorem}
\newtheorem{theorem}{Theorem}
\theoremstyle{lemma}
\newtheorem{lemma}{Lemma}
\theoremstyle{corollary}
\newtheorem{corollary}{Corollary}
\theoremstyle{remark}

\theoremstyle{definition}

\graphicspath{{figs/}}
\DeclareMathOperator{\sign}{sign}

\begin{document}
	
	\title{Private Inner Product Retrieval\\for Distributed Machine Learning}

	\author{
		\IEEEauthorblockN{Mohammad Hossein Mousavi$^{*}$, Mohammad Ali Maddah-Ali$^{\dagger}$, \\and Mahtab Mirmohseni$^{\dagger}$\vspace*{0.5em}\\
}		\IEEEauthorblockA{
				$^{*}$Department of Electrical Engineering, Sharif University of Technology, Tehran, Iran%
	\\$^{\dagger}$Nokia Bell Labs, Holmdel, NJ, USA}
			}

	\maketitle

	\begin{abstract}
		In this paper, we argue that in many basic algorithms for machine learning, including support vector machine (SVM) for classification, principal component analysis (PCA) for dimensionality reduction, and regression for dependency estimation, we need the inner products of the data samples, rather than the data samples themselves.
		
		Motivated by the above observation,  we introduce the problem of \emph {private inner product retrieval for distributed machine learning}, where we have a system including a database of some files, duplicated across some non-colluding servers. A user intends to retrieve a subset of specific size of the inner products of the data files with minimum communication load, without revealing any information about the identity of the requested subset.  For achievability, we use the algorithms for multi-message private information retrieval. For converse, we establish that as the length of the files becomes large, the set of all inner products converges to independent random variables with uniform distribution, and derive the rate of convergence. To prove that, we construct special dependencies among sequences of the sets of all inner products with different length, which forms a time-homogeneous irreducible Markov chain, without affecting the marginal distribution.  We show that this Markov chain has a uniform distribution as its unique stationary distribution, with rate of convergence dominated by the second largest eigenvalue of the transition probability matrix. This allows us to develop a converse, which converges to a tight bound in some cases, as the size of the files becomes large. While this converse is based on the one in multi-message private information retrieval due to the nature of retrieving inner products instead of data itself some changes are made to reach the desired result.
	\end{abstract}

	\section{Introduction}
	\label{sec:introduction}
	With the growth in data volume over recent years, the tasks of data storage and processing are often offloaded from in-house trusted systems to some external entities. Such distributed environments raise challenges, not experienced before.  One of the most important ones is privacy concern, which can have different interpretations. Based on the applications use-case, the private asset might be the training data, test data, and even the model parameters (the learning algorithm).  While the first two have been the subject of extensive research, from both computational cryptography and information-theoretic perspectives, the last one has been less understood.
	
	In the privacy of the machine learning \emph{algorithms}, the goal is to ensure the privacy of the parameters. Many different scenarios can be considered in which the parameters are in danger of breaching, and need to be addressed. Here, we focus on the case, where the learner must download some data samples from the servers to train the model.  In this case, the learner wants to keep the identity of this subset hidden from the servers. The reason is that in many cases, revealing the identity of the selected training samples would reveal considerable information about the intention of the learner, and can be used to guess the learning algorithm and calculate parameters of the model. For example, assume that learner downloads some training samples from a server to train a classification algorithm, say support vector machines (SVM). The server can easily guess that, and run the same algorithm, and gain full knowledge about the intention and the model.
	
	In this paper, we investigate the above privacy concern in a distributed setting, while our goal is to achieve privacy in a fundamental and information-theoretic level where no information is revealed about the algorithms to data owners. We argue that some of the most basic machine learning algorithms in different areas, including but not limited to SVM for classification, regression for relationship estimation, and principal component analysis (PCA) for dimensionality reduction, share an important feature in using sample data in their algorithm. To run these methods, the learner needs the inner products of the data files instead of the raw data. This can be particularly important when the length of input vectors is large compared to the number of data used for learning.

	On a separate line of research, the privacy in distributed settings, referred to as private information retrieval (PIR), is investigated.  In \cite{sun2017capacity}, the basic setup of PIR is studied, where the goal is to retrieve a file from a dataset, replicated in some non-colluding servers, without revealing its index. In particular, the capacity, as the infimum of the normalized download rate, is characterized. This is followed by \cite{sun2018capacity_1,sun2018capacity_2,banawan2018capacity,banawan2018noisy,wei2018fundamental} for different cases such as symmetric privacy, possibility of collusion among the servers,  and coded storage instead of uncoded replication of data files in servers. In particular, in \cite{banawan2018multi}, the multi-message PIR (MPIR) problem is studied, where the objective is to privately download a subset of files, instead of just one, and the capacity is approximately, and in some cases tightly,  characterized. The problem of retrieving a linear function of files from the servers, referred to as private computation (PC) or private function retrieval (PFR),  is investigated in~\cite{sun2018capacity} and~\cite{mirmohseni2018private}. In~\cite{obead2018capacity} the capacity for private linear computation in MDS coded databases is studied. Recently the new problem of retrieving a polynomial function of files from some servers has been introduced and discussed in~\cite{raviv2018private} and \cite{obead2019private} by using Lagrange encoding in coded databases.
	
	In this paper, we study a system, including a dataset of $K$ files, replicated across $N$ non-colluding servers. A user (learner) wishes to retrieve a subset of inner products out of all possible inner products of $K$ data files, without revealing the identity of the subset to each server. We prove that as the length of files, $L$, goes to infinity, the set of inner products of all data files (listed in the vector $X^{(L)}$) converges, in distribution, to a set of mutually independent uniform random variables.  To show that, we introduce some dependencies in the sequence of $X^{(L)}$, $L=1,2,\ldots$,  while keeping the marginal distribution of $X^{(L)}$ the same. Thanks to this dependency, we show that $\{ X^{(L)} \}_{L=1}^{\infty}$ forms a time-homogeneous irreducible Markov chain, with uniform distribution as its unique stationary distribution. Moreover, the rate of convergence is governed by the second largest eigenvalue $\lambda_2$ of the transition probability matrix, where $|\lambda_2| \leq 1$.
	This property motivates us to suggest MPIR as an achievable scheme.
	In addition, we rely on the above property to develop a converse which becomes tight in some case, as the length of files goes to infinity. While this converse is based on \cite{banawan2018multi}, a few changes are needed to be made to reach our goal. This is because of the difference in retrieving inner products instead of data files in \cite{banawan2018multi}. For example, the number of possible inner products cannot be any arbitrary integer which forces us to introduce an equivalent problem with arbitrary number of inner products in the process of reaching converse results.
	
	The organization of the paper is as follows: In Section~\ref{sec:motivation}, we discuss and motivate why retrieving the set of inner products are critical in machine learning. Next in Section~\ref{sec:problem}, we formally define the problem setting. We state our main results in Section~\ref{sec:main} and their proofs in Section~\ref{sec:Proof of Correctness}.
	
	\section{Background and Motivation}
	\label{sec:motivation}
	
	In what follows, we review some of the most basic machine learning algorithms, in three areas of classification, regression, and dimension reduction, and
	show that all three are based on the inner products of the samples, rather than the samples.
	
	1. \emph{Support vector machines (SVM)}: The SVM is one of the basic classification algorithms, where the goal is to correctly label the data files. This algorithm has many use cases such as face detection, bioinformatics (gene classifications), text categorization and etc. Here, we describe a simple case of SVM from~\cite[Page 63]{mohri2018foundations} and we discuss that knowing the inner products is enough to run the algorithm (instead of knowing the entire database).
	
	Consider an input alphabet $\mathcal{X}$ consisting of length $L$ vectors, a target output alphabet $\mathcal{Y} = \{-1, 1\}$ and a distribution $\mathcal{D}$ on $\mathcal{X}\times \mathcal{Y}$.
	The learner has $m$ training samples from $\mathcal{X} \times \mathcal{Y}$, denoted by $(\mathbf{x}_1, y_1), (\mathbf{x}_2, y_2), \ldots, (\mathbf{x}_m, y_m)$, drawn from $\mathcal{D}$. The goal here is to find function $h:\mathcal{X}\to\mathcal{Y}$ from hypothesis set $\mathcal{H}$, such that the following generalization error is minimized over $\mathcal{H}$:
	\begin{align}
	R_{\mathcal{D}}(h) = \Pr_{(\mathbf{x}, y) \sim \mathcal{D} }\{h(\mathbf{x})\neq y\}.
	\end{align}
	
	Although many different hypotheses sets exist, $\mathcal{H}$ can be chosen as described in \cite{mohri2018foundations} as a linear classifier defined as follows,
	\begin{align}
	\mathcal{H} = \{\mathbf{x} \mapsto \sign(\langle\mathbf{w},\mathbf{x}\rangle + b) | \mathbf{w} \in \mathbb{R}^{L} , b \in \mathbb{R}\} .
	\end{align}
	
	The solution to this problem boils down to solving the following convex optimization problem:
	\begin{align}
	&\min_{\mathbf{w},b} \dfrac{1}{2}||\mathbf{w}^2||\\
	\text{subject to: }&y_i(\langle\mathbf{w},\mathbf{x}_i\rangle + b) \geq 1, \forall i \in [1:m],
	\end{align}
	where for any integer $m$, $[1:m]$ denotes $\{1,\ldots,m\}$. This notation is used throughout this paper.
	The above problem can be solved by introducing Lagrange variables $\alpha_i \geq 0, i \in [1:m]$ for each constraint. Thus, the dual form of the constrained optimization problem is derived as following.
	\begin{align}
	&\max_{\alpha_i, i \in [1:m]} \sum_{i=1}^{m}\alpha_i - \dfrac{1}{2}\sum_{i,j=1}^{m}\alpha_i\alpha_j y_i y_j\langle\mathbf{x}_i,\mathbf{x}_j\rangle\label{eqn:SVM1}\\
	&\textrm{subject to: }\alpha_i \geq 0 \textrm{ and } \sum_{i=1}^{m}\alpha_i y_i = 0, \forall i \in [1:m].
	\end{align}
	
	Solving the dual problem on $\alpha_i, i \in [1:m]$, we have:
	\begin{align}
	\mathbf{w} &= \sum_{i=1}^{m}=\alpha_iy_i\mathbf{x}_i\;\;,\;\;b = y_i - \sum_{j=1}^{m}\alpha_jy_j\langle\mathbf{x}_j,\mathbf{x}_i\rangle.\label{eqn:SVM4}
	\end{align}		
	
	As is clear from \eqref{eqn:SVM1}-\eqref{eqn:SVM4}, in order to solve the main problem for $\mathbf{w},b$, we only need the inner products of samples and their labels to solve the dual problem for $\alpha_i$ and a linear combination of data samples to get $\mathbf{w}$~\footnote{To having a linear combination of the samples privately, we can use a scheme called \emph{private function retrieval}}.
	So, when the length of vectors $\mathbf{x}_i$, $L$, is large, retrieving inner products instead of raw samples is more efficient in a distributed learning setting.

2. \emph{Regression}: The regression algorithm predicts the real-valued label of a point by using a data set. Regression is a very common task in machine learning for approximately and closely deriving the relationship between variables. The regression is similar to continuous-label version of the classification, as opposed to the classification's discrete labels.
Many use cases can be considered for the regression algorithm, such as optimizing the price of products by learning the relation of price and the sale volume in different markets and analyzing the product sale drivers such as distribution methods in markets.
Here, we first describe a simple regression problem from~\cite[Page 245]{mohri2018foundations} and show in order to solve this problem we only need the inner products as opposed to retrieve all data files.

Similar to SVM, consider an input alphabet $\mathcal{X}$ consisting of vectors of length $L$ and a distribution $\mathcal{D}$ on $\mathcal{X}\times \mathcal{Y}$.  The learner has $m$ training samples from $\mathcal{X} \times \mathcal{Y}$, denoted by $(\mathbf{x}_1, y_1), (\mathbf{x}_2, y_2), \ldots, (\mathbf{x}_m, y_m)$, drawn from $\mathcal{D}$.
The difference is that the target output alphabet $\mathcal{Y}$ can be a continuous space. Since the labels are real numbers, the learner is not able to predict them precisely. So, a loss function is considered to show the distance between the label and the predicted value.
	
	Now, we discuss a simple linear regression problem. Similar to SVM, the hypothesis set is as follows.
	\begin{align}
	\mathcal{H} = \{\mathbf{x} \mapsto \langle\mathbf{w},\mathbf{x}\rangle + b | \mathbf{w} \in \mathbb{R}^{L} , b \in \mathbb{R}\}.
	\end{align}
	The loss here is empirical mean squared error. So, the optimization problem is as follows,
	\begin{align}
	\min_{\mathbf{w},b}\dfrac{1}{m}\sum_{i=1}^{m}(\langle\mathbf{w},\mathbf{x}_i\rangle + b - y_i)^2,
	\end{align}
	which can be written in a simpler form as:
	\begin{align}
	\min_{\tilde{\mathbf{w}}}F(\tilde{\mathbf{w}}) = \dfrac{1}{m}||\mathbf{X}^\top\tilde{\mathbf{w}} - \mathbf{y}||^2,
	\end{align}
	where $\mathbf{X} = \begin{bmatrix} \mathbf{x_1} &... &\mathbf{x_m}\\ 1 &... &1 \end{bmatrix}$, $\tilde{\mathbf{w}} = \begin{bmatrix} w_1 \\\vdots \\w_L\\ 1 \end{bmatrix}$ and  $\mathbf{y} = \begin{bmatrix} y_1 \\\vdots \\y_m \end{bmatrix}$. It is clear that the objective functions is convex and reaches its optimum value in $\nabla F(\tilde{\mathbf{w}}) = 0$. So, we have:
	\begin{align}
	\dfrac{2}{m}\mathbf{X}(\mathbf{X}^\top\tilde{\mathbf{w}} - \mathbf{y}) = 0 \Leftrightarrow \mathbf{X}\mathbf{X}^\top\tilde{\mathbf{w}} = \mathbf{X}\mathbf{y}.
	\end{align}
	
	Now, if $\mathbf{X}\mathbf{X}^\top$ is invertible, we can calculate $\tilde{\mathbf{w}}$. Otherwise, we replace the inverse with pseudo-inverse.
	\begin{align}
	\tilde{\mathbf{w}} = \begin{cases} (\mathbf{X}\mathbf{X}^\top)^{-1}\mathbf{X}\mathbf{y} &\text{if } \mathbf{X}\mathbf{X}^\top \text{ is invertible}\\ (\mathbf{X}\mathbf{X}^\top)^\dagger\mathbf{X}\mathbf{y} &\text{otherwise}\end{cases}.
	\end{align}
	
	It can be easily shown that the above result can be rewritten as below.
	\begin{align}
	\tilde{\mathbf{w}} = \begin{cases} \mathbf{X}(\mathbf{X}^\top\mathbf{X})^{-1}\mathbf{y} &\text{if } \mathbf{X}^\top\mathbf{X} \text{ is invertible}\\ \mathbf{X}(\mathbf{X}^\top\mathbf{X})^\dagger\mathbf{y} &\text{otherwise}\end{cases}.
	\end{align}
	
	As seen, the solution only needs inner products ($\mathbf{X}^\top\mathbf{X}$) and a linear combination of data files ($\tilde{\mathbf{w}} = \mathbf{X}\mathbf{a},$ $\mathbf{a} = (\mathbf{X}^\top\mathbf{X})^{-1}\mathbf{y}$) and not all data files. If the length of data vectors, $L$, is large, downloading all data files needs much more resource.
	
	3. \emph{Principal component analysis (PCA)}: The purpose of this algorithm is to reduce the dimensionality of data with large vector length, so that its most important features can be better analyzed. The reason is that sometimes the generalization ability of method decreases with the increase in dimension of data. The following example is from~\cite[Page 324]{shalev2014understanding}.
	
	Consider the $m$ vectors of length $L$, $\mathbf{x}_1,...,\mathbf{x}_m$, as data files. The goal here is to reduce the dimensionality of these vectors using linear transformation. To do this, we define a matrix $\mathbf{W}\in \mathbb{R}^{d\times L}$ where $d < L$. We also have a mapping $\mathbf{x}\mapsto\mathbf{Wx}$, whose output is the lower dimensionality representation of data. Then a second matrix $\mathbf{U}\in \mathbb{R}^{L\times d}$ is defined to recover $\mathbf{x}$. This means that if $\mathbf{y} = \mathbf{Wx}$ is the reduced representation, then the $\tilde{\mathbf{x}} = \mathbf{Uy}$ is the recovered data. Minimizing the magnitude of empirical distance between the original data and the recovered data is the goal of PCA.
	\begin{align}
	\arg\min_{\mathbf{W},\mathbf{U}}\sum_{i=1}^{m}||\mathbf{x}_i - \mathbf{UW}\mathbf{x}_i||^2.
	\end{align}
	It is shown in \cite{shalev2014understanding} that $\mathbf{U}^\top=\mathbf{W}$ and this problem can be rewritten as follows.
	\begin{align}
	\arg\min_{\mathbf{U}}\sum_{i=1}^{m}||\mathbf{x}_i& - \mathbf{U}\mathbf{U}^\top\mathbf{x}_i||^2,\\
	\text{subject to: } &\mathbf{U}^\top\mathbf{U} = \mathbf{I},
	\end{align}	
	where $\mathbf{I}$ is the identity matrix. According to Theorem 23.2 in~\cite[Page 325]{shalev2014understanding} the solution for above problem is to calculate $\mathbf{u}_1,...,\mathbf{u}_d$ which are eigenvectors of matrix $\mathbf{A} = \sum_{i=1}^{m}\mathbf{X}\mathbf{X}^\top$ ($\mathbf{X} = [\mathbf{x_1} ... \mathbf{x_m}]$) corresponding to $d$ largest eigenvalues of the matrix. The solution is $\mathbf{U} = [\mathbf{u}_1 ... \mathbf{u}_d]$. 
	
	If the dimension of the original vectors is too large ($L\gg m$), then we can rewrite the answer. We define $\mathbf{B} = \mathbf{X}^\top\mathbf{X}$. Let $\mathbf{u}$ be an eigenvector of matrix $\mathbf{B}$ (so $\mathbf{B}\mathbf{u} = \lambda\mathbf{u}$). This means that we have $\mathbf{X}^\top\mathbf{X}\mathbf{u} = \lambda\mathbf{u}$ and thus,
	\begin{align}
	\mathbf{X}\mathbf{X}^\top\mathbf{X}\mathbf{u} = \lambda\mathbf{X}\mathbf{u} \Rightarrow \mathbf{A}\mathbf{X}\mathbf{u} = \lambda\mathbf{X}\mathbf{u}.
	\end{align}
	
	Therefore, if $\mathbf{u}$ is an eigenvector of $\mathbf{B}$, corresponding to eigenvalue $\lambda$, then $\mathbf{Xu}$ is an eigenvector of matrix $\mathbf{A}$, corresponding to the same eigenvalue. So in PCA, when vector length $L$ is large, it is simpler to calculate the matrix $\mathbf{X}^\top\mathbf{X}$ that is matrix of inner products of original data.
	Then, the eigenvectors of this matrix corresponding to its $d$ largest eigenvalues are enough.
	
	These three algorithms make clear that methods using inner products of data files are important and common tasks of machine learning. Thus, retrieving the inner products privately from the servers is an important step in machine learning privacy.
	
\section{Problem Statement}
\label{sec:problem}
Consider a set of  $K$ data files, $W_1, \ldots, W_K$, for some integer $K$, where files are selected independently and uniformly at random from a finite field
$\mathbb{F}(q^L)$, for some integer $L$. Thus,
\begin{align}
H(W_1, W_2, \ldots, W_K)=LK \log(q).
\end{align}
Files can be represented in the vector form as
\begin{align}
&W_k = (w_{k1},...,w_{kL})^\top\nonumber\\
&w_{k \ell }\in\ \mathbb{F}(q), \textrm{for} \ k \in [1:K],  \ell \in [1:L].
\end{align}

We assume that files are replicated in $N$ non-colluding servers, for some integer $N$.
We define $\mathcal{X}^{(L)}$,  as the set of the inner product of all pairs of data files,
\begin{align}
\mathcal{X}^{(L)} = \{\langle W_i,W_j\rangle,\ \forall i,j \in [1:K]\}.
\end{align}
Also, we define $\mathcal{T}$ as index of inner products as follows,
\begin{align}
\mathcal{T}=\{\{i,j\},\ \forall i,j\in\{1,2,...,K\}\}.
\end{align}
Note that each member of $\mathcal{T}$ corresponds to an inner product in set $\mathcal{X}^{(L)}$, i.e., $\{i,j\}\in\mathcal{T} \iff \langle W_i,W_j\rangle \in \mathcal{X}^{(L)}$.

A user wishes to retrieve a subset of size $P \in \mathbb{N}$ of inner products. More precisely, the user chooses a set  $\mathcal{P}$, where  $\mathcal{P} \subseteq \mathcal{T}$, and $|\mathcal{P}|=P$, and entreats to know $\mathcal{X}^{(L)}_\mathcal{P}$, defined as
\begin{align}
&\mathcal{X}^{(L)}_{\mathcal{P}}=\{\langle W_i,W_j\rangle,  \forall \{i,j\}\in\mathcal{P}\}.
\end{align}
The cardinality $P$ of $\mathcal{P}$ is known to all servers.
The user wishes to retrieve $\mathcal{X}^{(L)}_{\mathcal{P}}$ while ensuring privacy of $\mathcal{P}$ from each server.

In order to retrieve these inner products user creates queries $Q_1^{[\mathcal{P}]},...,Q_N^{[\mathcal{P}]}$ and sends $Q_n^{[\mathcal{P}]}$ to server $n$, through an error-free secure link. In response, server $n$, responds with $A_n^{[\mathcal{P}]}$. Since user has no knowledge of files,
\begin{align}\label{eqn:indQ}
I(W_1,...,W_K;Q_1^{[\mathcal{P}]},...,Q_N^{[\mathcal{P}]}) = 0.
\end{align}
The answer of server $n$, $n\in [1:N]$,  is a function of query sent to that server
and the set of data files available there, thus	
\begin{align}
H(A_n^{[\mathcal{P}]}|W_1,...,W_K,Q_n^{[\mathcal{P}]}) = 0 \label{reliC}.
\end{align}
Also, $A_{n_1:n_2}$ denotes set $\{A_{n_1},A_{n_1+1},...,A_{n_2}\}$.
The queries and answers must satisfy two conditions:

{\bf (i) Correctness Condition:} This condition states that by having all queries and answers from servers, the user can calculate inner products indexed by the set $\mathcal{P}$. Equivalently,
\begin{align}
H(\mathcal{X}^{(L)}_\mathcal{P}|\ A_{1:N}^{[\mathcal{P}]},Q_{1:N}^{[\mathcal{P}]}) = 0 \label{corC}.
\end{align}

%
{\bf (ii) Privacy Condition:} In order to satisfy privacy,  regardless of what set $\mathcal{P}$ is chosen, query and answer for each server must be identically distributed, i.e., $\forall \mathcal{P}_1, \mathcal{P}_2 \subseteq \mathcal{T}$ , $|\mathcal{P}_1|=|\mathcal{P}_2|=P$, we must have,
\begin{align}
\label{privacy}
\!\!\!\!(Q_n^{[\mathcal{P}_1]}, A_n^{[\mathcal{P}_1]}, W_1,...,W_K)\!\sim\!(Q_n^{[\mathcal{P}_2]}, A_n^{[\mathcal{P}_2]}, W_1,...,W_K).
\end{align}
For an achievable scheme, satisfying~\eqref{corC} and \eqref{privacy},  we define the retrieval rate $R(P,L)$, as the ratio between information of the inner products in $\mathcal{X}^{(L)}_\mathcal{P}$ and total downloading cost to retrieve the inner products $\mathcal{X}^{(L)}_\mathcal{P}$, minimized over all possible requests $\mathcal{P} \subseteq \mathcal{T}$,  $|\mathcal{P}|=P$, i.e.,
\begin{align}
R(P,L) = \min_{\mathcal{P} \subseteq \mathcal{T},  |\mathcal{P}|=P  }\dfrac{H(\mathcal{X}^{(L)}_\mathcal{P})}{\sum_{n=1}^NH(A_n^{[\mathcal{P}]})}.
\end{align}

The capacity is the supremum of all achievable $R(P,L)$.
%
%
%

\section{Preliminary}
In order to proceed we need to review the results of MPIR problem in \cite{banawan2018multi}. Consider a  system, including $K$ data files, replicated in $N$ noncoluding  servers. Each data file is chosen independently and uniformly at  random from the finite field $\mathbb{F}(q^L)$. A user wishes to retrieve a subset indexed by $\mathcal{P} \subseteq [1:K]$ of data files, ensuring the privacy of $\mathcal{P}$. Assume $|\mathcal{P}| = P$, where $P$  is known publicly. Rate is defined as information of subset of data files indexed by $\mathcal{P}$ over download cost,  and the capacity $C_{\mathsf{MPIR}}$ is defined as the supremum over all rates in privacy preserving schemes. Then we have~\cite{banawan2018multi},
\begin{align}
\underline{R}_{\mathsf{MPIR}}(K,P,N) \leq {C}_{\mathsf{MPIR}} \leq  \overline{R}_{\mathsf{MPIR}}(K,P,N)\label{eqn:MPIR},
\end{align}
where for $\frac{K}{P} \leq 2$, we have
\begin{align}
\dfrac{1}{\overline{R}_{\mathsf{MPIR}}(K,P,N)}=\dfrac{1}{\underline{R}_{\mathsf{MPIR}}(K,P,N)}=1+\dfrac{K - P}{PN},\nonumber
\end{align}
and for $\frac{K}{P} \geq 2$, we have
\begin{align}
\dfrac{1}{\overline{R}_{\mathsf{MPIR}}(K,P,N)} =
\sum_{i=0}^{\text{\tiny{\(\left\lfloor \frac{K}{P} \right\rfloor - 1\)}}}\dfrac{1}{N^i}
+ \left(\frac{K}{P}- \left\lfloor \frac{K}{P} \right\rfloor\right)\dfrac{1}{N^{\text{\tiny{\(\left\lfloor \frac{K}{P} \right\rfloor\)}}}},\nonumber
\end{align}
and $\dfrac{1}{\underline{R}_{\mathsf{MPIR}}(K,P,N)}$ is equal to
\begin{align}
\dfrac{\sum_{i=1}^{P}\beta_ir_i^{\text{\tiny{\(K-P\)}}}\left[\left(1+\dfrac{1}{r_i}\right)^{\text{\tiny{\(K\)}}} - \left(1+\dfrac{1}{r_i}\right)^{\text{\tiny{\(K-P\)}}}\right]}{\sum_{i=1}^{P}\beta_ir_i^{\text{\tiny{\(K-P\)}}}\left[\left(1+\dfrac{1}{r_i}\right)^{\text{\tiny{\(K\)}}} - 1\right]},\nonumber
\end{align}
where $r_i$ is defined as $	r_i = \dfrac{e^{\hat{j}2\pi(i-1)/P}}{N^{1/P} - e^{\hat{j}2\pi(i-1)/P}}, i\in [1:P],$
and $\hat{j} = \sqrt{-1}$. In addition,  $\beta_i, i\in [1:P]$, is the solution of the set of linear equations $\sum_{i=1}^{P}\beta_ir_i^{-P}=(N-1)^{K - P}$ and $\sum_{i=1}^{P}\beta_ir_i^{-k}=0$,  $k\in[1:P-1]$.

\section{Main Results}
\label{sec:main}
The main result is stated in the following theorem.

\begin{theorem}\label{mpl2}
	For a system with $K$ files in $\mathbb{F}(q^L)$ and $N$ servers, where the user is interested in a subset of size $P$ of inner products,  we have
	\begin{align}
	\dfrac{1}{\underline{R}_{\mathsf{MPIR}}(K(K+1)/2,P,N)} - &O(\lambda_2^{L-1}) < \dfrac{1}{C}  \label{eqn:main}\\
	&\leq \dfrac{1}{\overline{R}_{\mathsf{MPIR}}(K(K+1)/2,P,N)},\nonumber
	\end{align}
	where $\lambda_2$ is a constant independent of $L$ and $|\lambda_2| < 1$.
\end{theorem}

\begin{corollary}
	\label{cor:mpl2}
	If $\dfrac{K(K+1)}{2P} \leq 2$, then we have
	\begin{align}
	\lim_{L\to\infty}\frac{1}{C} =1+\dfrac{K(K+1) - 2P}{2PN}.
	\end{align}
\end{corollary}
\begin{corollary}
	\label{cor:mpg2}
	If $\dfrac{K(K+1)}{2P} \in \mathbb{N}$,  then  we have
	\begin{align}
	\lim_{L\to\infty} \dfrac{1}{C} = 1+ \dfrac{1}{N} + ... + \dfrac{1}{N}^{\text{\tiny{\(\dfrac{K(K+1)}{2P}-1\)}}}.
	\end{align}
\end{corollary}
The proof can be found in the next section.  Assuming $q$ is large enough, for achievability, we use the scheme of MPIR. For converse, we prove that as $L$ goes to infinity, entries of $\mathcal{X}^{(L)}$ converges to a set of independent random variables with uniform distribution, with the rate of convergence dominated by a constant $\lambda_2$, $|\lambda_2|\le 1$. For large $L$, in some cases, the achievable rate and converse match. In other cases, these two are very close.

	\section{Proof}
		\label{sec:Proof of Correctness}

We sort the elements of set $\mathcal{X}^{(L)}$ in a vector $X^{(L)} \in \mathbb{F}^{K(K+1)/2}(q)$, such that $\langle W_i,W_j\rangle$ in $X^{(L)}$ comes before $\langle W_k,W_l\rangle$ if $i<k$ or $i=k$ and $j<l$. Likewise, we sort the elements of $\mathcal{X}^{(L)}_\mathcal{P}$ in a vector  $X^{(L)}_{\mathcal{P}}$.
%
%

	In this section, we provide the proof for Theorem~\ref{mpl2}.
	
First we show that as $L\to\infty$, the distribution of $X^{(L)}$ converges to a uniform distribution over $\mathbb{F}^{K(K+1)/2}(q)$: 	
\begin{align}
\forall \mathbf{y} \in\mathbb{F}^{K(K+1)/2}(q): \ \lim_{L \rightarrow \infty } \Pr \{X^{(L)} = \mathbf{y} \} = \dfrac{1}{q^{K(K+1)/2}}\label{limUni}.
\end{align}
Indeed, we increase $L$ by one, and show that the distribution of $X^{(L)}$ over $\mathbb{F}^{K(K+1)/2}(q)$ becomes closer to a uniform distribution.
In addition, we derive the rate of convergence.

	
	
	Let us denote the $q^{K(K+1)/2}$ members of set $\mathbb{F}^{K(K+1)/2 }(q)$ by $\mathbf{y}_1\ ...\ \mathbf{y}_{q^{K(K+1)/2}}$, i.e.,
	\begin{align}
	\mathbb{F}^{K(K+1)/2 }(q)  = \{\mathbf{y}_1\ ...\ \mathbf{y}_{q^{K(K+1)/2}}\}
	\end{align}
	
	We denote the probability mass function of $X^{(L)}$ over  $\mathbb{F}^{K(K+1)/2 }(q)$ by $\mathbf{p}^{(L)} \in [0,1]^{ q^{K(K+1)/2}}$, i.e.
	\begin{align}
	\mathbf{p}^{(L)} = (p^{(L)}_1,...,p^{(L)}_{q^{K(K+1)/2}})^\top \in [0,1]^{ q^{K(K+1)/2}}
	\end{align}
	where
	\begin{align}
	\label{piL}			
	p^{(L)}_i = \Pr \{X^{(L)} = \mathbf{y}_i\}, \  i\in[1:q^{K(K+1)/2}].			
	\end{align}	
	Apparently,
	\begin{align}			
	\sum_{i=1}^{q^{K(K+1)/2} }p^{(L)}_i = 1.		
	\end{align}	
	
	Our goal is to investigate how $\mathbf{p}^{(L)}$ changes, as we increase $L$ to $L+1$.  Let
	\begin{align}
	W_i^{(L)} = (w_{i1},...,w_{iL})^{\top},  \  i \in [1:K].
	\end{align}
	Without loss of generality, we assume that
	\begin{align}
	W_i^{(L+1)} \triangleq  (w_{i1},...,w_{iL},w_{i(L+1)})^{\top} \ i \in [1:K],
	\end{align}
	where $w_{i(L+1)}$ is selected uniformly at random from $\mathbb{F}(q)$. We note that by this construction $X^{(L)}$ and $X^{(L+1)}$ become correlated. However, the distribution of $X^{(L+1)}$ is still the same as it was discussed in the problem formulation but this correlation allows us to derive the converging distribution.

	\begin{lemma}\label{ismarkov}
		The sequence ${\{ X^{(L)} \}}_{L=1}^{\infty}$ forms a Markov chain with a time-homogeneous transition probability  $ \mathbf{M} \in \mathbb{R}^{q^{K(K+1)/2}\times q^{K(K+1)/2}}$, i.e.
		\begin{align}
		\mathbf{p}^{(L+1)} = \mathbf{M}\mathbf{p}^{(L)},
		\end{align}
		where
		\begin{align}
		[M]_{i,j} = \Pr\{ \Delta^{(L,L+1)} = \mathbf{y}_i - \mathbf{y}_j \}, \ \forall i,j \in [1:q^{K(K+1)/2}].
		\end{align}		
	\end{lemma}
	
	\begin{proof}
		
		Defining the  data files as above, then  we have,
		\begin{align}
		\langle W_i^{(L+1)},W_j^{(L+1)} \rangle= \langle W_i^{(L)},W_j^{(L)}\rangle + w_{i(L+1)}w_{j(L+1)}, \ \forall i, j \in [1:K].
		\end{align}
		Thus for the vector of inner products $X^{(L)}$, we also can write,
		\begin{align}
		X^{(L+1)} = X^{(L)} + \Delta^{(L,L+1)}  \ \ \
		\end{align}
		where
		\begin{align}
		\Delta^{(L,L+1)}   = (w_{1(L+1)}w_{1(L+1)}, w_{1(L+1)}w_{2(L+1)},...,w_{K(L+1)}w_{K(L+1)})^{\top} \in \mathbb{F}^{K(K+1)/2}(q) .
		\end{align}
		
		Because of the way we constructed $W_i^{(L+1)}$ from $W_i^{(L)}$, for $i=[1:K]$,   it is apparent that $ \Delta^{(L+1)}$ is independent of data files $W_i^{(L)}$, $i \in [1:K]$, and irrespective of $L$. We have
		\begin{align}
		\Pr\{X^{(L+1)} = \mathbf{y}_i\}=\sum_{j \in q^{K(K+1)/2}} \Pr\{X^{(L)} = \mathbf{y}_j\}.\Pr\{\Delta^{(L,L+1)} = \mathbf{y}_i - \mathbf{y}_j\}
		\end{align}	
		Thus from \eqref{piL}, we can rewrite the above equation as 		
		\begin{align}
		\mathbf{p}^{(L+1)} = \mathbf{M}\mathbf{p}^{(L)},
		\end{align}
		where $ \mathbf{M} \in \mathbb{R}^{q^{K(K+1)/2}\times q^{K(K+1)/2}}$ is a constant matrix, with entry $(i,j)$ be equal to
		\begin{align}
		[M]_{i,j} = \Pr\{ \Delta^{(L,L+1)} = \mathbf{y}_i - \mathbf{y}_j \}, \ \forall i,j \in [1:q^{K(K+1)/2}].
		\end{align}
		We note that $ \mathbf{M}$ is constant and independent of $L$.
	\end{proof}

	To show that the limit in \eqref{limUni} exists, in the following lemma, we guarantee that the Markov chain has steady distribution.
	
	\begin{lemma}
		\label{lem:positive}
		Markov chain formed by the sequence ${\{ X^{(L)} \}}_{L=1}^{\infty}$ is irreducible.
	\end{lemma}
	\begin{proof}
		In order to prove lemma we show that there exists some $\Gamma \in \mathbb{N}$, such that  $[\mathbf{M}^\Gamma]_{i,j} > 0$,  $\forall i,j \in [1:q^{K(K+1)/2}] $. This means it is possible to get to any state from any state in this chain or equivalently this chain is irreducible.
		We note that for any integer $\Gamma$
		\begin{align}
		X^{(L+\Gamma)} = X^{(L)} + \Delta^{(L,L+\Gamma)},
		\end{align}
		where
		\begin{align}
		\Delta^{(L,L+\Gamma)} = (\sum_{\gamma=1}^{\Gamma}w_{1(L+\gamma)}w_{1(L+\gamma)}, \sum_{\gamma=1}^{\Gamma}w_{1(L+\gamma)}w_{2(L+\gamma)},...,\sum_{\gamma=1}^{\Gamma}w_{K(L+\gamma)}w_{K(L+\gamma)})^{\top}.
		\end{align}
		One can see that
		\begin{align}
		\Pr\{ \Delta^{(L,L+\Gamma)} = \mathbf{y}_i - \mathbf{y}_j \} = [\mathbf{M}^\Gamma]_{i,j}, \ \forall i,j \in [1:K] \label{Midxprob},
		\end{align}
		$[\mathbf{M}^\Gamma]_{i,j}$ denotes entry $(i,j)$ of matrix $\mathbf{M}^\Gamma$.

		This lemma is equivalent to claim that there exists some $\Gamma \in \mathbb{N}$ such  that every realization of $\Delta^{(L,L+\Gamma)}$ in $\mathbb{F}^{K(K+1)/2}(q)$ is possible with some positive probability. Notice that the following relationship holds,
		\begin{align}
		\Delta^{(L,L+\Gamma)} = \sum_{\gamma=1}^{\Gamma}\Delta^{(L+\gamma-1,L+\gamma)}
		\end{align}
		It is obvious that $\Delta^{(L+\gamma-1,L+\gamma)}$, $\gamma=1,\ldots, \Gamma$,  are mutually independent.
		The reason is that $\Delta^{(L+\gamma-1,L+\gamma)}$ is only dependent of $w_{k(L+\gamma)}, k\in [1:K]$.
		
		We first show that for $\Gamma= 5$ every vector in $\mathbb{F}^{K(K+1)/2}(q)$ with only one non-zero element is a probable (has a  positive probability) realization of $\Delta^{(L,L+5)}$. In other words, we show that, for any	$\mathbf{y}\in\mathbb{F}^{K(K+1)/2}(q)$, where $\mathbf{y}(e) \neq 0$ and $\mathbf{y}(i) =0$, $\forall i  \in [1:q^{K(K+1)/2}]\backslash \{e \}$,  for some $e
		\in [1:q^{K(K+1)/2}]$,  then $\Pr\{\Delta^{(L,L+5)} = \mathbf{y}\}>0$. Let us assume
		\begin{align}
		\mathbf{y}(e) = a, \textrm{for some}  \ a \in \mathbb{F}(q) \backslash \{ 0 \}.
		\end{align}
		
		%
		%
		We know that by definition  $\Delta^{(L,L+5)} (e) = \sum_{\gamma=1}^{5}w_{i_e(L+\gamma)}.w_{j_e(L+\gamma)}$, for some $i_e,j_e \in [1:K]$. Here, we consider  two cases for values of $i_e,j_e$.
		\begin{itemize}
			\item[{\bf Case (I)}]: In this case $i_e=j_e$.  In other words,  $\Delta^{(L,L+5)}(e) = \sum_{\gamma=1}^{5}w_{i_e(L+\gamma)}.w_{i_e(L+\gamma)}$, for some $i_e \in [1:K]$
			
			From \cite[Page 66 ]{newman1972integral}, we have
			\begin{align}
			\forall a\in \mathbb{F}(q),\exists s,t \in \mathbb{F}(q) : a = s^2+t^2\label{sum2},
			\end{align}
			Therefore one possible case that can create such  $\mathbf{y}$ is as follows:
			\begin{align}
			w_{r(L+\gamma)} = \begin{cases}
			t & r=i_e=j_e,\gamma=1 \\
			s & r=i_e=j_e,\gamma=2 \\
			0 & o.w.
			\end{cases}
			\end{align}
			Clearly this case has positive probability and therefore 	$\Pr\{ \Delta^{(L,L+\Gamma)} = \mathbf{y} \} > 0$.

			%
			%

			\item[{\bf Case (II):}] In the case $i_e \neq j_e$.  In other words,  $\Delta^{(L,L+5)}(e) = \sum_{\gamma=1}^{5}w_{i_e(L+\gamma)}.w_{j_e(L+\gamma)}$, for some $i_e, j_e \in [1:K]$,  $i_e \neq j_e$.

			We have (see~\cite[Page 66]{newman1972integral})
			\begin{align}
			\exists s_1,s_2,t_1,t_2 \in \mathbb{F}_q  : -a^2 = s_1^2+t_1^2 \text{  and} -1 = s_2^2+t_2^2
			\end{align}
			Therefore one possible case that can create such  $\mathbf{y}$ is as follows:
			\begin{align}
			\label{caseII}
			w_{r(L+\gamma)} = \begin{cases}
			a & r=i_e,\gamma=1 \\
			1 & r=j_e,\gamma=1 \\
			s_1 & r=i_e,\gamma=2 \\
			t_1 & r=i_e,\gamma=3 \\
			s_2 & r=j_e,\gamma=4 \\
			t_2 & r=j_e,\gamma=5 \\
			0 & o.w.
			\end{cases}
			\end{align}
			In particular, one can verify that
			%
			\begin{align}
			\Delta^{(L,L+5)}(e)=\sum_{\gamma=1}^{\Gamma}w_{i_e(L+\gamma)}w_{j_e(L+\gamma)} = a\times 1 + s_1\times 0  + t_1\times 0 + 0\times s_2 + 0\times t_2 = a.
			\end{align}
			In addition
			\begin{align}
			\sum_{\gamma=1}^{\Gamma}w_{i_e(L+\gamma)}w_{i_e(L+\gamma)} = a^2 + s_1^2 + t_1^2 + 0 + 0 = a^2 - a^2 = 0,
			\end{align}
			\begin{align}
			\sum_{\gamma=1}^{\Gamma}w_{j_e(L+\gamma)}w_{j_e(L+\gamma)} = 1^2 + 0 + 0 + s_2^2 + t_2^2 = 1^2 - 1^2 = 0.
			\end{align}
			Other entries of $\Delta^{(L,L+5)}$ are zero trivially.
			
			Since the probability of \eqref{caseII} is not zero, therefore in this case also $\Pr\{ \Delta^{(L,L+\Gamma)} = \mathbf{y} \} > 0$.
			
			%
			%
			
		\end{itemize}
		
		From these two cases above,  we can say every vector with one non-zero element is a probable (with positive probability) realization of $\Delta^{(L,L+5)}$.
		We now show that every vector in $\mathbb{F}^{K(K+1)/2}(q)$ is a possible realization with positive probability for $\Delta^{(L,L+\Gamma)}$ when $\Gamma \geq 5K(K+1)/2$. First we write $\Delta^{(L,L+5K(K+1)/2)}$ as,
		\begin{align}
		\Delta^{(L,L+5K(K+1)/2)} = \sum_{\gamma=1}^{K(K+1)/2}\Delta^{(L+5(\gamma-1),L+5(\gamma))}\label{sumdelta}.
		\end{align}
		
		Let $\mathbf{y}\in \mathbb{F}^{K(K+1)/2}(q)$ be an arbitrary vector. To show that $\mathbf{y}$ is a possible realization of $\Delta^{(L,L+\Gamma)}$ with non-zero probability,  we first define $\mathbf{y}^{(i)},i\in[1:K(K+1)/2]$,  as follows,
		\begin{align}
		\mathbf{y}^{(i)}(i) & = \mathbf{y}(i),\\
		\mathbf{y}^{(i)}(j)  & = 0, \  j \in[1:K(K+1)/2] \backslash \{ i \}.
		\end{align}
		This means $\mathbf{y}^{(i)}$ is zero in every index except in index $i$ where its value is $\mathbf{y}(i)$. We can see that,
		\begin{align}
		\mathbf{y}=\sum_{i=1}^{K(K+1)/2}\mathbf{y}^{(i)}.\label{sumy}
		\end{align}
		
		By construction, $\mathbf{y}^{(i)}$ is a vector that has at most one non-zero element, thus it is a probable realization for $\Delta^{(L+5(i-1),L+5(i))}$. Now if $\Delta^{(L+5(i-1),L+5(i))}=\mathbf{y}^{(i)}$, ,$\forall i\in[1:K(K+1)/2]$ which is possible with positive probability then because of \eqref{sumy} and \eqref{sumdelta}, we know $\Delta^{(L,L+5K(K+1)/2)}=\mathbf{y}$, therefore $\Pr\{\Delta^{(L,L+5K(K+1)/2)}=\mathbf{y}\} > 0$. Also because of \eqref{Midxprob} every element in the matrix $\mathbf{M}^{5K(K+1)/2}$ is positive.
	\end{proof}

	\begin{corollary}
		Markov sequence ${\{ X^{(L)} \}}_{L=1}^{\infty}$  has a steady state.	
	\end{corollary}
	
	\begin{proof}
		Markov sequence  ${\{ X^{(L)} \}}_{L=1}^{\infty}$ has a unique steady state if there exist an integer $\Gamma$ that $\mathbf{M}^{\Gamma}$ has an all positive row~~\cite[Page 176]{ash1990information}, as it is proved in Lemma~\ref{lem:positive}.
	\end{proof}		
	
\begin{lemma}\label{markovchain}
	As $L\rightarrow \infty$,  Markov chain  ${\{ X^{(L)} \}}_{L=1}^{\infty}$ converges to a random vector with uniform distribution over $\mathbb{F}^{K(K+1)/2}(q)$.
\end{lemma}

	\begin{proof}
		
		It is known that if a Markov chain has
		steady state, its stationary distribution is equal to its steady
		state probabilities~~\cite[Page 174]{ash1990information}. We use this fact to find that steady state. As obtained, we know $[\mathbf{M}]_{i,j} = \Pr\{\Delta^ {(L,L+1)} = \mathbf{y}_j - \mathbf{y}_i\}$. It is easy to see that for any $i$,
		the set $\{ \mathbf{y}_i - \mathbf{y}_j, j \in [1: q^{K(K+1)/2}] \}$ is equal to $\mathbb{F}^{K(K+1)/2 }(q)$. Thus,
		\begin{align}
		\label{Mfull}
		\sum_{j=1}^{q^{K(K+1)/2}} [\mathbf{M}]_{i,j} = 1, \ \forall i \in [1:q^{K(K+1)/2}].
		\end{align}
		Let $\boldsymbol{\pi} = (1/q^{K(K+1)/2},...,1/q^{K(K+1)/2})^\top$.
		It is easy to see that due to~\eqref{Mfull}, $\mathbf{M} \boldsymbol{\pi} =\boldsymbol{\pi}$.
		%
		%
		Thus uniform distribution is stationary state probability of this Markov chain.
	\end{proof}

	\begin{lemma}\label{rate}
		Let $\mathbf{p}^{(L)}  \in [0,1]^{ q^{K(K+1)/2}}$ denote the PMF of $X^{(L)}$ over $\mathbb{F}^{K(K+1)/2 }(q)$. Then, $ \|\mathbf{p}^{(L)} -\boldsymbol{\pi} \|_{\infty} = O(\lambda_2^{L - 1})$, where $\lambda_2$ is the second largest eigenvalue (absolute value of eigenvalue) of $ \mathbf{M}$ and $|\lambda_2| <1$.
	\end{lemma}
	
	We show that,
	\begin{align}
	\mathbf{p}^{(L)} + O(\lambda_2^{L-1})\mathbf{1} = \boldsymbol{\pi}
	\end{align}
	where $\boldsymbol{\pi} = (\dfrac{1}{q^{K(K+1)/2}},...,\dfrac{1}{q^{K(K+1)/2}})^\top$ and $\mathbf{1} = (1,...,1)^\top \in \mathbb{R}^{K(K+1)/2}$.
	
	From Lemma~\ref{ismarkov}, we know that $\{ X^{(L)} \}_{L=1}^{\infty}$ forms a Markov chain  with transition matrix $\mathbf{M}$.
	Thus, the PMF of
	$X^{(L)}$, denoted by $\mathbf{p}^{(L)}$, is equal to
	\begin{align}
	\mathbf{p}^{(L)} = \mathbf{M} \mathbf{p}^{(1)},
	\end{align}
	where $\mathbf{p}^{(1)}$ is the PMF of $X^{(1)}$.  Also from Lemma~\ref{markovchain}, we know that this Markov chain has a steady state. Thus for the eigenvalue of
	transition matrix $\mathbf{M}$, we have,
	\begin{align}
	\label{lambda}
	|\lambda_{q^{K(K+1)/2}}| \leq  ... \leq |\lambda_2| < |\lambda_1| = 1.
	\end{align}
	As stated, matrix $\mathbf{M}$ and thus its eigenvalues are independent of $L$.
	
	Let $\boldsymbol{\pi},  \mathbf{v}_2, \ldots,  \mathbf{v}_{q^{K(K+1)/2}} \in \mathbb{R}^{q^{K(K+1)/2}} $ denote the right eigenvectors corresponding to the eigenvalues $\lambda_1, \lambda_2, \ldots, \lambda_{q^{K(K+1)/2}}$, respectively.
	We note that $\boldsymbol{\pi},  \mathbf{v}_2, \ldots,  \mathbf{v}_{q^{K(K+1)/2}}$ forms an orthogonal bases for $\mathbb{R}^{q^{K(K+1)/2}}$. Thus, we can expand  $\mathbf{p}^{(1)}$ as
	\begin{align}
	\mathbf{p}^{(1)}  = \alpha_1 \boldsymbol{\pi}  + \alpha_2 \mathbf{v}_2 +...+  \alpha_{q^{K(K+1)/2}} \mathbf{v}_{q^{K(K+1)/2}}\label{P1topi},
	\end{align}
	for some $\alpha_i$, $i \in [1: {q^{K(K+1)/2}} ]$.

	Therefore we can write,
	\begin{align}
	\mathbf{p}^{(L)} = \mathbf{M} \mathbf{p}^{(1)} = \lambda_1^{L-1} \alpha_1 \boldsymbol{\pi}  + \lambda_2^{L-1} \alpha_2 \mathbf{v}_2 +...+  \lambda_{q^{K(K+1)/2}}^{L-1} \alpha_{q^{K(K+1)/2}} \mathbf{v}_{q^{K(K+1)/2}} \label{PpiTemp}.
	\end{align}
	 From the fact that $\lim_{L \rightarrow \infty} \mathbf{p}^{(L)} = \boldsymbol{\pi}$ and also \eqref{lambda} we know when $L\to\infty$ every term in \eqref{PpiTemp} diminishes except $\lambda_1^{L-1}\alpha_1\boldsymbol{\pi}$ (where $\lambda_1 = 1$) which should be equal to $\boldsymbol{\pi}$. Thus we can rewrite \eqref{PpiTemp} as,
	 \begin{align}
	 \mathbf{p}^{(L)} =  \boldsymbol{\pi}  + \lambda_2^{L-1} \alpha_2 \mathbf{v}_2 +...+  \lambda_{q^{K(K+1)/2}}^{L-1} \alpha_{q^{K(K+1)/2}} \mathbf{v}_{q^{K(K+1)/2}} \label{Ppi}.
	 \end{align}
	 
	 Thus for every element of $\mathbf{p}^{(L)}$ from \eqref{Ppi}, we have,
	\begin{align}
	\mathbf{p}^{(L)}(i) &= \boldsymbol{\pi}(i) + \lambda_2^{L-1}\alpha_2 \mathbf{v}_2(i) +...+  \lambda_{q^{K(K+1)/2}}^{L-1}\alpha_{q^{K(K+1)/2}} \mathbf{v}_{q^{K(K+1)/2}}(i) \\
	&\leq \boldsymbol{\pi}(i) + |\lambda_2^{L-1}|(\sum_{t=2}^{q^{K(K+1)/2}}|\alpha_t \mathbf{v}_t(i)|) = \boldsymbol{\pi}(i) + O(\lambda_2^{L-1}).
	\end{align}

	\begin{lemma}\label{entropy}
		Entropy of set $\mathcal{X}^{(L)}_\mathcal{P}$ of inner products when $|\mathcal{P}| = P$  satisfies,
		\begin{align}
		H(\mathcal{X}^{(L)}_\mathcal{P}) \geq P\log(q) - O(\lambda_2^{L-1})
		\end{align}
		
	\end{lemma}
	\begin{proof}
		This lemma gives a lower bound on entropy of inner products indexed by a subset $\mathcal{P}$. We first calculate probability distribution of $X^{(L)}_{\mathcal{P}}$ over $\mathbb{F}^P(q)$. For the set
		$\mathcal{P}$, we define the vector of indices $\boldsymbol{\tau}_\mathcal{P}$, such that
		$X^{(L)}_{\mathcal{P}}(i)=X^{(L)}( \tau_\mathcal{P}(i))$, for $i=1, \ldots, P$.  In addition, for a $\mathbf{z} \in \mathbb{F}^P(q)$, we define
		\begin{align}
		\mathcal{S}_{\mathbf{z}} \triangleq \left\{\mathbf{x}|\mathbf{x}\in\mathbb{F}^{K(K+1)/2}(q) , \mathbf{x}(\tau_{\mathcal{P}}(i)) =  \mathbf{z}(i), \forall i \in [1:P] \right\}.
		\end{align}
		It is easy to see that $|\mathcal{S}_{\mathbf{z}}| = q^{K(K+1)/2 - P}$.
		Now we can calculate the probability distribution of $X^{(L)}_{\mathcal{P}}$ over $\mathbb{F}^P(q)$.
		\begin{align}
		\Pr\{X^{(L)}_{\mathcal{P}} = \mathbf{z}\} &= \sum_{ \mathbf{x}\in \mathcal{S}_{\mathbf{z}}  }\Pr\{X^{(L)} = \mathbf{x}\} \\&= \sum_{ \mathbf{x}\in \mathcal{S}_{\mathbf{z}}  }  \dfrac{1}{q^{K(K+1)/2}} + O(\lambda_2^{L-1})\label{ref1setP} \\&= \dfrac{q^{K(K+1)/2 - P}}{q^{K(K+1)/2}}+ O(\lambda_2^{L-1}) = \dfrac{1}{q^P} + O(\lambda_2^{L-1}),
		\end{align}
		where
		(\ref{ref1setP}) is the result of Lemma \ref{rate}.
		
		For entropy of inner products in set $\mathcal{X}^{(L)}_{\mathcal{P}}$, we can write,		
		\begin{align}
		H(\mathcal{X}^{(L)}_{\mathcal{P}}) = H(X^{(L)}_{\mathcal{P}}) &=-\sum_{\mathbf{z}\in \mathbb{F}^P(q)}\Pr\{X^{(L)}_{\mathcal{P}} = \mathbf{z}\} \log{(\Pr\{X^{(L)}_{\mathcal{P}} = \mathbf{z}\})} \\
		&\geq -\sum_{\mathbf{z}\in \mathbb{F}^P(q)}\left(\dfrac{1}{q^P} + O(\lambda_2^{L-1})\right)\left(\log{\dfrac{1}{q^P} + O(\lambda_2^{L-1})}\right) \\
		&= -\sum_{\mathbf{z}\in \mathbb{F}^P(q)}\dfrac{1}{q^P} \log{\dfrac{1}{q^P}} - O(\lambda_2^{L-1}) \\&= P \log(q) - O(\lambda_2^{L-1})
		\end{align}
	\end{proof}
	
	Now that we have the result of lemma \ref{entropy} and by employing MPIR achievable scheme on inner products we are able to achieve result of theorem \ref{mpl2}.
	
	\subsection{Proof of Theorem \ref{mpl2}}
	This proof is similar to the proof in \cite{banawan2018multi} changed to match the current problem setting. To show the limits of capacity, this proof needs to be split in two parts, achievability and converse.
	
\emph{Achievability:} We use the same achievability scheme that is used in MPIR problem. In this case as the problem setting states, user wants to privately retrieve all inner products included in $\mathcal{X}^{(L)}_\mathcal{P}$ without revealing the identity of $\mathcal{P}$. Now we treat every inner product like an entry data file in MPIR and run the proposed scheme on them. Notice that by running MPIR scheme, inner products indexed by the mentioned subset are privately retrieved and to servers all subsets of size $P$ are equiprobable so privacy constraint is met and we can say,
	\begin{align}
	C \geq \underline{R}_{\mathsf{MPIR}}(K(K+1)/2,P,N) \Rightarrow \dfrac{1}{C} \leq \dfrac{1}{\underline{R}_{\mathsf{MPIR}}(K(K+1)/2,P,N)}\label{ach}.
	\end{align}
	Here the number of data files in MPIR is the number of inner products, i.e. $K(K+1)/2$.

\emph{Converse:} In order to prove the converse for Theorem \ref{mpl2} (i.e., to derive an upper bound on the capacity), we use Lemma \ref{entropy} from which we have a lower bound on the entropy of an arbitrary subset of inner products $\mathcal{P}$ as,
\begin{align}\label{eqn:con_ent1}
H(\mathcal{X}_{\mathcal{P}}) \geq P \log(q) - O(\lambda_2^{L-1}).
\end{align}

To continue with the proof, we consider the problem in two cases, $\dfrac{K(K+1)}{2P} \leq 2$ and $\dfrac{K(K+1)}{2P} > 2$.\\
\textbf{Case 1}: $\dfrac{K(K+1)}{2P} \leq 2$.

Remember that the user sends queries $Q_n^{[\mathcal{P}]}, n \in [1:N]$, to servers and receives answers $A_n^{[\mathcal{P}]},n\in[1:N]$. We define the set of all possible queries as
\begin{align}
	\mathcal{Q} \triangleq \{Q_n^{[\mathcal{P}]}| \mathcal{P}\subseteq\mathcal{T},\ n \in [1:N]\},
\end{align}
and also the set of answers from servers $n_1$ to $n_2$ as,
\begin{align}
	A_{n_1:n_2} \triangleq \{A_{n_1}^{[\mathcal{P}]},A_{n_1+1}^{[\mathcal{P}]},...,A_{n_2}^{[\mathcal{P}]}\}.
\end{align}
Note that the number of all possible $\mathcal{P}$ is $\beta = \binom{K(K+1)/2}{P}$.

We can assume symmetry in the scheme across data files and servers queries and answers. Even if the scheme is asymmetric we can replicate scheme for every permutation of data bases and servers and create a symmetric scheme.

Since the queries and answers are independent of the desired set because of \eqref{privacy}, we fix the answers of server 1 to be (same as the MPIR),
\begin{align}\label{answersym}
	A_1^{[\mathcal{P}]} = A_1, \forall \mathcal{P}\subseteq\mathcal{T},|\mathcal{P}|=P.
\end{align}

The proof of the following lemma is similar to the proof of \cite[Lemma~1]{banawan2018multi}.
\begin{lemma}\label{symL}
For any $\mathcal{S}\subseteq\mathcal{T}$ and $\mathcal{X}^{(L)}_{\mathcal{S}}=\{\langle W_i,W_j\rangle|\{i,j\}\in\mathcal{S}\}$, we have,
\begin{align}
	H(A_n^{[\mathcal{P}]}|\mathcal{X}^{(L)}_{\mathcal{S}}, \mathcal{Q}) & = H(A_1^{[\mathcal{P}]}|\mathcal{X}^{(L)}_{\mathcal{S}}, \mathcal{Q}),\\
	H(A_1|\mathcal{Q}) & = H(A_n^{[\mathcal{P}]}|\mathcal{Q})\label{priRes}.
\end{align}
\end{lemma}
\eqref{priRes} is the result of symmetry assumption and \eqref{privacy}.

First, we derive a lower bound on the entropy of inner products appeared in $A_1^{[\mathcal{P}]}$, that are not in $\mathcal{X}^{(L)}_\mathcal{P}$.
\begin{lemma}\label{MPIR_lemma_rem_info}
For the problem stated in Section~\ref{sec:problem}, with $P\geq\dfrac{K(K+1)}{2}$, the following lower bound on the conditional entropy of $A_1^{[\mathcal{P}]}$ holds.
	\begin{align}
	H(A_1^{[\mathcal{P}]}|\mathcal{X}^{(L)}_{\mathcal{P}}, \mathcal{Q}) > \dfrac{\dfrac{K(K+1)}{2}- P}{N}\log(q) - O(\lambda_2^{L-1}).
	\end{align}
\end{lemma}
\begin{proof}
Here we define $\bar{\mathcal{P}}\subseteq\mathcal{T}$ to be the set with size $|\bar{\mathcal{P}}|=P$ such that $\mathcal{P}\cup\bar{\mathcal{P}}=\mathcal{T}$.

Such a set $\bar{\mathcal{P}}$ exists, since $\mathcal{P}$ has more than the half of all inner products (due to $P \geq K(K+1)/4$). Now, we use \eqref{eqn:con_ent1} for $\mathcal{X}^{(L)}$ to write the following.
\begin{align}
\dfrac{K(K+1)}{2} \log(q) - P\log(q) - O(\lambda_2^{L-1})
&\leq   H(\mathcal{X}^{(L)}) - H(\mathcal{X}^{(L)}_\mathcal{P})\\
&= H(\mathcal{X}^{(L)}\setminus\mathcal{X}^{(L)}_\mathcal{P}|\mathcal{X}^{(L)}_\mathcal{P},\mathcal{Q})\\
&= H(\mathcal{X}^{(L)}\setminus\mathcal{X}^{(L)}_\mathcal{P}|\mathcal{X}^{(L)}_\mathcal{P},\mathcal{Q}) - H(\mathcal{X}^{(L)}\setminus\mathcal{X}^{(L)}_\mathcal{P}|A_{1:N}^{[\bar{\mathcal{P}}]},\mathcal{X}^{(L)}_\mathcal{P},\mathcal{Q}) \label{ref1remfo}\\
&= I(\mathcal{X}^{(L)}\setminus\mathcal{X}^{(L)}_\mathcal{P};A_{1:N}^{[\bar{\mathcal{P}}]}|\mathcal{X}^{(L)}_\mathcal{P},\mathcal{Q})\\
&= H(A_{1:N}^{[\bar{\mathcal{P}}]}|\mathcal{X}^{(L)}_\mathcal{P},\mathcal{Q})\label{ref2remfo}\\
&\leq \sum_{n=1}^{N}H(A_n^{[\bar{\mathcal{P}}]}|\mathcal{X}^{(L)}_\mathcal{P},\mathcal{Q})\\
&=NH(A_1|\mathcal{X}^{(L)}_\mathcal{P},\mathcal{Q}),\label{ref3remfo}
\end{align}
where \eqref{ref1remfo} follows from \eqref{corC} noting that $\mathcal{X}^{(L)}\setminus \mathcal{X}^{(L)}_\mathcal{P} \subseteq \mathcal{X}^{(L)}_{\bar{\mathcal{P}}}$; 
\eqref{ref2remfo} is true because of \eqref{reliC}; and \eqref{ref3remfo} is the result of Lemma~\ref{symL}.
\end{proof}

Now, we proceed to the converse following the approach of \cite{banawan2018multi} as,
\begin{align}
	\dfrac{K(K+1)}{2}\log(q)  - O(\lambda_2^{L-1})&\leq H(\mathcal{X}^{(L)})\\
	&= H(\mathcal{X}^{(L)}|\mathcal{Q})\label{ref4main1}\\
	&= H(\mathcal{X}^{(L)}|\mathcal{Q}) - H(\mathcal{X}^{(L)}|A_{1:N}^{[\mathcal{P}_1]},...,A_{1:N}^{[\mathcal{P}_\beta]},\mathcal{Q})\label{ref1main1}\\
	&= I(\mathcal{X}^{(L)};A_{1:N}^{[\mathcal{P}_1]},...,A_{1:N}^{[\mathcal{P}_\beta]}|\mathcal{Q})\\
	&= H(A_{1:N}^{[\mathcal{P}_1]},...,A_{1:N}^{[\mathcal{P}_\beta]}|\mathcal{Q})\label{ref2main1}\\
	&= H(A_1,A_{2:N}^{[\mathcal{P}_1]},...,A_{2:N}^{[\mathcal{P}_\beta]}|\mathcal{Q})\label{ref3main1}\\
	&= H(A_1,A_{2:N}^{[\mathcal{P}_1]}|\mathcal{Q}) + H(A_{2:N}^{[\mathcal{P}_2]},...,A_{2:N}^{[\mathcal{P}_\beta]}|A_1,A_{2:N}^{[\mathcal{P}_1]},\mathcal{Q})\\
	&= H(A_1,A_{2:N}^{[\mathcal{P}_1]}|\mathcal{Q}) + H(A_{2:N}^{[\mathcal{P}_2]},...,A_{2:N}^{[\mathcal{P}_\beta]}|A_1,A_{2:N}^{[\mathcal{P}_1]},\mathcal{X}^{(L)}_{\mathcal{P}_1},\mathcal{Q})\label{ref5main1}\\
	&\leq \sum_{n = 1}^{N}H(A_n^{[\mathcal{P}_1]}|\mathcal{Q}) + H(A_{2:N}^{[\mathcal{P}_2]},...,A_{2:N}^{[\mathcal{P}_\beta]}|A_1,\mathcal{X}^{(L)}_{\mathcal{P}_1},\mathcal{Q})\\
	&=  \sum_{n = 1}^{N}H(A_n^{[\mathcal{P}_1]}|\mathcal{Q}) + H(A_{1:N}^{[\mathcal{P}_2]},...,A_{1:N}^{[\mathcal{P}_\beta]}|\mathcal{X}^{(L)}_{\mathcal{P}_1},\mathcal{Q}) - H(A_1|\mathcal{X}^{(L)}_{\mathcal{P}_1}, \mathcal{Q}),\label{ref_last_main1}
	\end{align}
where \eqref{ref4main1} follows \eqref{eqn:indQ}; \eqref{ref1main1} follows from \eqref{corC} and the fact that $\mathcal{P}_1,\ldots,\mathcal{P}_\beta$ are distinct sets that cover all possible inner products and thus every inner product can be decoded from $A_{1:N}^{[\mathcal{P}_1]},...,A_{1:N}^{[\mathcal{P}_\beta]}$; \eqref{ref2main1} holds thanks to \eqref{reliC}; \eqref{ref3main1} is true because of \eqref{answersym}; and \eqref{ref5main1} is due to \eqref{corC}.

We can see that,
\begin{align}
	H(A_{1:N}^{[\mathcal{P}_2]},...,A_{1:N}^{[\mathcal{P}_\beta]}|\mathcal{X}^{(L)}_\mathcal{P},\mathcal{Q}) &\leq H(A_{1:N}^{[\mathcal{P}_2]},...,A_{1:N}^{[\mathcal{P}_\beta]},\mathcal{X}^{(L)}|\mathcal{X}^{(L)}_\mathcal{P},\mathcal{Q})=H(\mathcal{X}^{(L)}|\mathcal{X}^{(L)}_\mathcal{P},\mathcal{Q})\label{reliCres}\\
	&= H(\mathcal{X}^{(L)}) - H(\mathcal{X}^{(L)}_\mathcal{P})  \\
&\leq  \dfrac{K(K+1)}{2}\log(q)  - P\log(q) + O(\lambda_2^{L-1})\label{ref_last2_main1},
\end{align}
where \eqref{reliCres} is true because of \eqref{reliC} and the last inequality follows from \eqref{eqn:con_ent1}.
Now, combining the result of Lemma~\ref{MPIR_lemma_rem_info} with \eqref{ref_last_main1} and \eqref{ref_last2_main1} results in:
	\begin{align}
	\dfrac{K(K+1)}{2}\log(q)  - O(\lambda_2^{L-1})\leq &\sum_{n = 1}^{N}H(A_n^{[\mathcal{P}_1]}|\mathcal{Q}) +  \dfrac{K(K+1)}{2}\log(q) - P\log(q) + O(\lambda_2^{L-1}) \nonumber\\
&- \dfrac{K(K+1)- 2P}{2N}\log(q) + O(\lambda_2^{L-1}),
	\end{align}
which can be written as,
\begin{align}
P\log(q) + \dfrac{K(K+1)- 2P}{2N}\log(q) - O(\lambda_2^{L-1})\leq \sum_{n = 1}^{N}H(A_n^{[\mathcal{P}_1]}|\mathcal{Q}),
\end{align}
and further is simplified as,
\begin{align}
\sum_{n = 1}^{N}H(A_n^{[\mathcal{P}_1]}|\mathcal{Q})&\geq
P\log(q)\left(1 + \dfrac{K(K+1)- 2P}{2PN} - O(\lambda_2^{L-1})\right)\\
&\geq
H(\mathcal{X}^{(L)}_\mathcal{P})\left(1 + \dfrac{K(K+1)- 2P}{2PN} - O(\lambda_2^{L-1})\right),
\end{align}
which completes the converse for Case 1 as:
\begin{align}
1 + \dfrac{K(K+1)- 2P}{2PN} - O(\lambda_2^{L-1}) \leq \dfrac{1}{C}.\label{case1Con}
\end{align}
\\
\textbf{Case 2}: $\dfrac{K(K+1)}{2P} > 2$.

Here, similar to the proof of MPIR, we create an inductive relation and use the result of Case~1 as the base induction step. In order to continue with the induction proof we need to introduce a slightly different problem formulation (we call as modified problem). Assume that the user wishes to retrieve inner products indexed by the set $\mathcal{P}_1$ while ensuring privacy over $\forall\mathcal{P}\subseteq\mathcal{T}_1\subseteq\mathcal{T}$ where $|\mathcal{T}_1|=T_1$ and $\mathcal{T}_1$ is fixed and known to all. We denote the query set in this problem by $\hat{\mathcal{Q}}$ and the answer from $n$-th server by $\hat{A}^{[\mathcal{P}]}_n$. Number of all possible sets $\mathcal{P}$ is $\hat{\beta}=\binom{T_1}{P}$.

Note that for $\mathcal{T}_1=\mathcal{T}$, the modified problem reduces to the original problem (of Section~\ref{sec:problem}). The reason for introducing the modified problem is that $|\mathcal{T}|$ in the original problem is equal to $|\mathcal{T}|=K(K+1)/2$ (that is $|\mathcal{T}|$ cannot take any arbitrary integer). However, for the inductive step, we need the number of inner products to take any integer.

It is easy to see that the result of case 1 is still true for the modified problem. Hence, when $T_1/P \leq 2$ we have,
\begin{align}
NH(\hat{A}_1|\hat{\mathcal{Q}})&\geq
P\log(q)\left(1 + \dfrac{T_1 - P}{PN} - O(\lambda_2^{L-1})\right).\label{inducBase}
\end{align}

We use this result as our induction base step. Now we proceed to prove that for case $T_1/P > 2$ we have,

\begin{align}
NH(\hat{A}_1|\hat{\mathcal{Q}}) \geq P\log(q)\left[\sum_{i=0}^{{\left\lfloor \text{\tiny{\(\dfrac{T_1}{P}\)}} \right\rfloor - 1}}\dfrac{1}{N^i} + \left(\dfrac{T_1}{P} - \left\lfloor \dfrac{T_1}{P} \right\rfloor\right)\dfrac{1}{N^{\left\lfloor \text{\tiny{\(\dfrac{T_1}{P}\)}} \right\rfloor }}-O(\lambda_2^{L-1})\right].\label{inducresult}
\end{align}
We achieve \eqref{inducresult} by induction with \eqref{inducBase} as induction base. First, we provide a lemma to derive an upper bound on the remaining information in answer $\hat{A}_{2:N}^{[\mathcal{P}_2]}$ conditioned on inner products indexed by the set $\mathcal{P}_1\subseteq\mathcal{T}_1$ when $\mathcal{P}_1$ and $\mathcal{P}_2 $ have the same size but have no similar inner product in their corresponding sets.

\begin{lemma}\label{interference}
For two sets $\mathcal{P}_1\subseteq\mathcal{T}_1$ and $\mathcal{P}_2\subseteq\mathcal{T}_1$, if $|\mathcal{P}_1| = |\mathcal{P}_2| = P$ and $\mathcal{P}_1\cap\mathcal{P}_2 = \phi$, then the following inequality holds.
	\begin{align}
	H(\hat{A}_{2:N}^{[\mathcal{P}_2]}|\mathcal{X}^{(L)}_{\mathcal{P}_1},\hat{\mathcal{Q}})\leq(N-1)[NH(\hat{A}_1|\hat{\mathcal{Q}}) - P\log(q)] + O(\lambda_2^{L-1}).
	\end{align}
\end{lemma}
\begin{proof}
	\begin{align}
	H(\hat{A}_{2:N}^{[\mathcal{P}_2]}|\mathcal{X}^{(L)}_{\mathcal{P}_1},\hat{\mathcal{Q}})
	&\leq \sum_{n=2}^{N}H(\hat{A}_{n}^{[\mathcal{P}_2]}|\mathcal{X}^{(L)}_{\mathcal{P}_1},\hat{\mathcal{Q}})\\
	&\leq \sum_{n=2}^{N}H(\hat{A}_{1:n-1}^{[\mathcal{P}_1]},\hat{A}_{n}^{[\mathcal{P}_2]},\hat{A}_{n+1:N}^{[\mathcal{P}_1]}|\mathcal{X}^{(L)}_{\mathcal{P}_1},\hat{\mathcal{Q}})\\
	&=\sum_{n=2}^{N}H(\hat{A}_{1:n-1}^{[\mathcal{P}_1]},\hat{A}_{n}^{[\mathcal{P}_2]},\hat{A}_{n+1:N}^{[\mathcal{P}_1]},\mathcal{X}^{(L)}_{\mathcal{P}_1}|\hat{\mathcal{Q}}) - H(\mathcal{X}^{(L)}_{\mathcal{P}_1}|\hat{\mathcal{Q}})\\
	&=\sum_{n=2}^{N}H(\hat{A}_{1:n-1}^{[\mathcal{P}_1]},\hat{A}_{n}^{[\mathcal{P}_2]},\hat{A}_{n+1:N}^{[\mathcal{P}_1]}|\hat{\mathcal{Q}}) + H(\mathcal{X}^{(L)}_{\mathcal{P}_1}|\hat{A}_{1:n-1}^{[\mathcal{P}_1]},\hat{A}_{n}^{[\mathcal{P}_2]},\hat{A}_{n+1:N}^{[\mathcal{P}_1]},\hat{\mathcal{Q}})- H(\mathcal{X}^{(L)}_{\mathcal{P}_1}) \label{mpg21}\\
	&\leq \sum_{n=2}^{N}NH(\hat{A}_1|\hat{\mathcal{Q}}) - P\log(q) + O(\lambda_2^{L-1})\label{mpg22}\\
	&=(N-1)[NH(\hat{A}_1|\hat{\mathcal{Q}}) - P\log(q)] + O(\lambda_2^{L-1}),
	\end{align}
where \eqref{mpg21} holds because of independence of queries and data files; and \eqref{mpg22} follows from \eqref{eqn:con_ent1}, the symmetry across the servers and the fact that $\mathcal{X}^{(L)}_{\mathcal{P}_1}$ can be calculated from $\hat{A}_{1:n-1}^{[\mathcal{P}_1]},\hat{A}_{n}^{[\mathcal{P}_2]},\hat{A}_{n+1:N}^{[\mathcal{P}_1]}$ and queries which is the result of \eqref{corC}.
\end{proof}

Now we construct the inductive step, similar to MPIR but tailored to our setting. Assume that two subsets $\mathcal{P}_1$ and $\mathcal{P}_2$ are chosen from all subsets $\mathcal{P}_i\subseteq\mathcal{T}_1,i\in[1:\hat{\beta}],|\mathcal{P}| = P$ such that $\mathcal{P}_1\cap\mathcal{P}_2=\phi$. We see,
\begin{align}
&T_1\log(q) - O(\lambda_2^{L-1})\leq H(\hat{A}_1,\hat{A}_{2:N}^{[\mathcal{P}_1]},...,\hat{A}_{2:N}^{[\mathcal{P}_{\hat{\beta}}]}|\hat{\mathcal{Q}})\label{eqn:case2_1}\\
&=H(\hat{A}_1,\hat{A}_{2:N}^{[\mathcal{P}_1]}|\hat{\mathcal{Q}}) + H(\hat{A}_{2:N}^{[\mathcal{P}_2]}|\hat{A}_1,\hat{A}_{2:N}^{[\mathcal{P}_1]},\hat{\mathcal{Q}}) + H(\hat{A}_{2:N}^{[\mathcal{P}_3]},...,\hat{A}_{2:N}^{[\mathcal{P}_{\hat{\beta}}]}|\hat{A}_1,\hat{A}_{2:N}^{[\mathcal{P}_1]},\hat{A}_{2:N}^{[\mathcal{P}_2]},\hat{\mathcal{Q}})\\
&\leq NH(\hat{A}_1|\hat{\mathcal{Q}}) + H(\hat{A}_{2:N}^{[\mathcal{P}_2]}|\hat{A}_1,\hat{A}_{2:N}^{[\mathcal{P}_1]},\mathcal{X}^{(L)}_{\mathcal{P}_1},\hat{\mathcal{Q}}) + H(\hat{A}_{2:N}^{[\mathcal{P}_3]},...,\hat{A}_{2:N}^{[\mathcal{P}_{\hat{\beta}}]}|\hat{A}_1,\hat{A}_{2:N}^{[\mathcal{P}_1]},\hat{A}_{2:N}^{[\mathcal{P}_2]},\mathcal{X}^{(L)}_{\mathcal{P}_1},\mathcal{X}^{(L)}_{\mathcal{P}_2},\hat{\mathcal{Q}}) \label{mpg23}\\
&\leq NH(\hat{A}_1|\hat{\mathcal{Q}}) + H(\hat{A}_{2:N}^{[\mathcal{P}_2]}|\mathcal{X}^{(L)}_{\mathcal{P}_1},\hat{\mathcal{Q}}) + H(\hat{A}_{2:N}^{[\mathcal{P}_3]},...,\hat{A}_{2:N}^{[\mathcal{P}_{\hat{\beta}}]}|\hat{A}_1,\mathcal{X}^{(L)}_{\mathcal{P}_1},\mathcal{X}^{(L)}_{\mathcal{P}_2},\hat{\mathcal{Q}})\\
&= NH(\hat{A}_1|\hat{\mathcal{Q}}) + H(\hat{A}_{2:N}^{[\mathcal{P}_2]}|\mathcal{X}^{(L)}_{\mathcal{P}_1},\hat{\mathcal{Q}}) + H(\hat{A}_{1:N}^{[\mathcal{P}_3]},...,\hat{A}_{1:N}^{[\mathcal{P}_{\hat{\beta}}]}|\mathcal{X}^{(L)}_{\mathcal{P}_1},\mathcal{X}^{(L)}_{\mathcal{P}_2},\hat{\mathcal{Q}}) - H(\hat{A}_1|\mathcal{X}^{(L)}_{\mathcal{P}_1},\mathcal{X}^{(L)}_{\mathcal{P}_2},\hat{\mathcal{Q}})\\
&\leq NH(\hat{A}_1|\hat{\mathcal{Q}}) + H(\hat{A}_{2:N}^{[\mathcal{P}_2]}|\mathcal{X}^{(L)}_{\mathcal{P}_1},\hat{\mathcal{Q}}) + H(\mathcal{X}^{(L)}_{\mathcal{T}_1}|\mathcal{X}^{(L)}_{\mathcal{P}_1},\mathcal{X}^{(L)}_{\mathcal{P}_2}) - H(\hat{A}_1|\mathcal{X}^{(L)}_{\mathcal{P}_1},\mathcal{X}^{(L)}_{\mathcal{P}_2},\hat{\mathcal{Q}})\label{mpg24}\\
&= NH(\hat{A}_1|\hat{\mathcal{Q}}) + H(\hat{A}_{2:N}^{[\mathcal{P}_2]}|\mathcal{X}^{(L)}_{\mathcal{P}_1},\hat{\mathcal{Q}}) + H(\mathcal{X}^{(L)}_{\mathcal{T}_1})-H(\mathcal{X}^{(L)}_{\mathcal{P}_1},\mathcal{X}^{(L)}_{\mathcal{P}_2}) - H(\hat{A}_1|\mathcal{X}^{(L)}_{\mathcal{P}_1},\mathcal{X}^{(L)}_{\mathcal{P}_2},\hat{\mathcal{Q}})\\
&\leq NH(\hat{A}_1|\hat{\mathcal{Q}}) + H(\hat{A}_{2:N}^{[\mathcal{P}_2]}|\mathcal{X}^{(L)}_{\mathcal{P}_1},\hat{\mathcal{Q}}) + T_1\log(q)- 2P\log(q) + O(\lambda_2^{L-1})\nonumber\\
&\quad - H(\hat{A}_1|\mathcal{X}^{(L)}_{\mathcal{P}_1},\mathcal{X}^{(L)}_{\mathcal{P}_2},\hat{\mathcal{Q}})\label{eqn:case2_3}\\
&\leq NH(\hat{A}_1|\hat{\mathcal{Q}}) + (N-1)[NH(\hat{A}_1|\hat{\mathcal{Q}}) - P\log(q)]+ T_1\log(q) - 2P\log(q) + O(\lambda_2^{L-1})\nonumber\\
&\quad - H(\hat{A}_1|\mathcal{X}^{(L)}_{\mathcal{P}_1},\mathcal{X}^{(L)}_{\mathcal{P}_2},\hat{\mathcal{Q}}) ,\label{eqn:case2_2}
\end{align}
where \eqref{eqn:case2_1} follows from the fact that knowing $A_1,A_{2:N}^{[\mathcal{P}_1]},...,A_{2:N}^{[\mathcal{P}_{\hat{\beta}}]}$ one can obtain all inner products indexed by $\mathcal{T}_1$ (noting that it is always possible to fix the answers of one of the servers (to be independent of the desired set) as done in \eqref{answersym});
\eqref{mpg23} is true due to the symmetry across the servers and the fact that $\mathcal{X}^{(L)}_{\mathcal{P}_1}$ and $\mathcal{X}^{(L)}_{\mathcal{P}_2}$ have no more information if we have $A_1,A_{2:N}^{[\mathcal{P}_1]},A_{2:N}^{[\mathcal{P}_2]}$ and queries which is the result of \eqref{corC}; \eqref{mpg24} is true because of \eqref{reliC} similar to \eqref{reliCres}; and \eqref{eqn:case2_3} follows from \eqref{eqn:con_ent1}.


We rewrite \eqref{eqn:case2_2} as,
\begin{align}
N^2H(\hat{A}_1|\hat{\mathcal{Q}}) \geq (N + 1)P\log(q) + H(\hat{A}_1|\mathcal{X}^{(L)}_{\mathcal{P}_1},\mathcal{X}^{(L)}_{\mathcal{P}_2},\hat{\mathcal{Q}}) - O(\lambda_2^{L-1}),
\end{align}
which can also be written as,
\begin{align}
NH(\hat{A}_1|\hat{\mathcal{Q}}) \geq \left(1 + \dfrac{1}{N}\right)P\log(q) + \dfrac{1}{N}H(\hat{A}_1|\mathcal{X}^{(L)}_{\mathcal{P}_1},\mathcal{X}^{(L)}_{\mathcal{P}_2},\hat{\mathcal{Q}}) - O(\lambda_2^{L-1}). \label{inducstep}
\end{align}

Now, as mentioned in \cite{banawan2018multi}, $H(\hat{A}_1|\mathcal{X}^{(L)}_{\mathcal{P}_1},\mathcal{X}^{(L)}_{\mathcal{P}_2},\hat{\mathcal{Q}})$ (in the modified problem) is quite similar to the $H(\hat{A}_1|\hat{\mathcal{Q}})$ in an equivalent problem where the user wants to retrieve a subset of inner products indexed by $\mathcal{P}\subseteq\mathcal{T}_2$ and $\mathcal{T}_2 = \mathcal{T}_1\setminus(\mathcal{P}_1\cup\mathcal{P}_2)$ and thus $|\mathcal{T}_2|=T_1-2P$. The difference between the modified problem and its equivalent problem is that in the modified problem we have conditions on $\mathcal{X}^{(L)}_{\mathcal{P}_1},\mathcal{X}^{(L)}_{\mathcal{P}_2}$, while in the equivalent problem the conditions don't exist. These two problem would be completely equivalent if the inner products were mutually independent. However, by adding the conditions, one can obtain all the equations up to \eqref{inducstep} for the equivalent problem with a difference of $O(\lambda_2^{L-1})$.
Therefore, we can say $H(\hat{A}_1|\mathcal{X}^{(L)}_{\mathcal{P}_1},\mathcal{X}^{(L)}_{\mathcal{P}_2},\hat{\mathcal{Q}})$ in the modified problem with total number of inner products $|\mathcal{T}_1|=T_1$ is equal to $H(\hat{A}_1|\hat{\mathcal{Q}})$ in an equivalent problem with total number of inner products $|\mathcal{T}_1|=T_1-2P$. We use this result in our induction step.

Now, we start with the following induction hypothesis when the number of all inner products is $|\mathcal{T}_1|=T_1-2P+1$:
\begin{align}
NH(\hat{A}_1|\hat{\mathcal{Q}}) \geq &P\log(q)\left[\sum_{i=0}^{{\left\lfloor \text{\tiny{\(\dfrac{T_1-2P+1}{P}\)}} \right\rfloor - 1}}\dfrac{1}{N^i}\right]\nonumber\\
&+ P\log(q)\left[\left(\dfrac{T_1-2P+1}{P} - \left\lfloor \dfrac{T_1-2P+1}{P} \right\rfloor\right)\dfrac{1}{N^{\left\lfloor \text{\tiny{\(\dfrac{T_1-2P+1}{P}\)}} \right\rfloor }}-O(\lambda_2^{L-1})\right]\label{induchypo}
\end{align}

To complete the proof by induction we must show that \eqref{induchypo} holds for $|\mathcal{T}_1|=T_1+1$.


Using the equivalency of the modified and the equivalent problem, \eqref{induchypo} (which was written for the equivalent problem with $|\mathcal{T}_1|=T_1-2P+1$) is true for the modified problem if we substitute $H(\hat{A}_1|\mathcal{X}^{(L)}_{\mathcal{P}_1},\mathcal{X}^{(L)}_{\mathcal{P}_2},\hat{\mathcal{Q}})$ by $H(\hat{A}_1|\hat{\mathcal{Q}})$ (as well as $T_1-2P+1$ by  $T_1+1$) to obtain:
\begin{align}
&NH(\hat{A}_1|\mathcal{X}^{(L)}_{\mathcal{P}_1},\mathcal{X}^{(L)}_{\mathcal{P}_2},\hat{\mathcal{Q}}) \\
&\geq P\log(q)\left[\sum_{i=0}^{{\left\lfloor \text{\tiny{\(\dfrac{T_1-2P+1}{P}\)}} \right\rfloor - 1}}\dfrac{1}{N^i}\right]\nonumber\\
&+ P\log(q)\left[\left(\dfrac{T_1-2P+1}{P} - \left\lfloor \dfrac{T_1-2P+1}{P} \right\rfloor\right)\dfrac{1}{N^{\left\lfloor \text{\tiny{\(\dfrac{T_1-2P+1}{P}\)}} \right\rfloor }}-O(\lambda_2^{L-1})\right]\\
&= P\log(q)\left[\sum_{i=0}^{{\left\lfloor \text{\tiny{\(\dfrac{T_1+1}{P}\)}} \right\rfloor - 3}}\dfrac{1}{N^i} + \left(\dfrac{T_1+1}{P} - \left\lfloor \dfrac{T_1+1}{P} \right\rfloor\right)\dfrac{1}{N^{\left\lfloor \text{\tiny{\(\dfrac{T_1+1}{P}\)}} \right\rfloor - 2}}-O(\lambda_2^{L-1})\right]\label{eqn:case2_5}
\end{align}

Combining \eqref{inducstep} and \eqref{eqn:case2_5} results in,
\begin{align}
&NH(\hat{A}_1|\hat{\mathcal{Q}}) \geq \left(1 + \dfrac{1}{N}\right)P\log(q) \nonumber\\
&+ \dfrac{1}{N^2}P\log(q)\left[\sum_{i=0}^{{\left\lfloor \text{\tiny{\(\dfrac{T_1+1}{P}\)}} \right\rfloor - 3}}\dfrac{1}{N^i} + \left(\dfrac{T_1+1}{P} - \left\lfloor \dfrac{T_1+1}{P} \right\rfloor\right)\dfrac{1}{N^{\left\lfloor \text{\tiny{\(\dfrac{T_1+1}{P}\)}} \right\rfloor - 2}}-O(\lambda_2^{L-1})\right]\\
&=P\log(q)\left[\sum_{i=0}^{{\left\lfloor \text{\tiny{\(\dfrac{T_1+1}{P}\)}} \right\rfloor - 1}}\dfrac{1}{N^i} + \left(\dfrac{T_1+1}{P} - \left\lfloor \dfrac{T_1+1}{P} \right\rfloor\right)\dfrac{1}{N^{\left\lfloor \text{\tiny{\(\dfrac{T_1+1}{P}\)}} \right\rfloor}} - O(\lambda_2^{L-1})\right]\label{inducStep}
\end{align}

From induction hypothesis \eqref{induchypo}, we have proven inductive step in \eqref{inducStep} and thus \eqref{inducresult} is true. As mentioned before, when $\mathcal{T}_1=\mathcal{T}$ or equivalently $T_1 = K(K+1)/2$, the modified problem reduces to the original problem. Therefore \eqref{inducresult} holds for the original problem. Therefore, we have,

\begin{align}
\sum_{i=0}^{{\left\lfloor \text{\tiny{\(\dfrac{K(K+1)}{2P}\)}} \right\rfloor - 1}}\dfrac{1}{N^i} + \left(\dfrac{K(K+1)}{2P} - \left\lfloor \dfrac{K(K+1)}{2P} \right\rfloor\right)\dfrac{1}{N^{\left\lfloor \text{\tiny{\(\dfrac{K(K+1)}{2P}\)}} \right\rfloor}} - O(\lambda_2^{L-1})
&\leq \dfrac{NH(A_1|\mathcal{Q})}{P\log(q)} \\
&\leq \dfrac{NH(A_1|\mathcal{Q})}{H(\mathcal{X}^{(L)}_\mathcal{P})} \\
&\leq \dfrac{\sum_{n=1}^{N}H(A_n^{\mathcal{P}})}{H(\mathcal{X}^{(L)}_\mathcal{P})} \\
&\leq \dfrac{1}{C}\label{case2Con}
\end{align}

From the results of two cases \eqref{case1Con} and \eqref{case2Con} we can write,
\begin{align}
\dfrac{1}{\underline{R}_{\mathsf{MPIR}}(K(K+1)/2,P,N)} - O(\lambda_2^{L-1}) < \dfrac{1}{C}\label{con}.
\end{align}

Comparing  \eqref{eqn:MPIR} with the results of converse, \eqref{con}, and achievability, \eqref{ach}, the proof of Theorem \ref{mpl2} is complete.
%
%
%

	\bibliography{References}
	\bibliographystyle{ieeetr}
	
\end{document}